\documentclass[submission,copyright,creativecommons]{eptcs}
\usepackage{amsmath,amsthm,amssymb,qtree}
\providecommand{\event}{GandALF 2020}
\usepackage{underscore}           

\title{A Game Theoretical Semantics for Logics of Nonsense}
\author{Can Ba\c{s}kent
\institute{Department of Computer Science, \\ Middlesex University, London, UK}
\email{c.baskent@mdx.ac.uk}
}
\def\titlerunning{Game Semantics for Logics of Nonsense}
\def\authorrunning{C. Ba\c{s}kent}

\newtheorem{thm}{Theorem} [section]     
\newtheorem{cor}[thm]{Corollary} 
 
\newtheorem{prop}[thm]{Proposition} 
\newtheorem{rem}[thm]{Remark}

\theoremstyle{definition}
\newtheorem{dfn}[thm]{Definition}
\newtheorem{ex}[thm]{Example}
\newtheorem{obs}[thm]{Observation}

\begin{document}
\maketitle

\begin{abstract}
Logics of non-sense allow a third truth value to express propositions that are \emph{nonsense}. These logics are ideal formalisms to understand how errors are handled in programs and how they propagate throughout the programs once they appear. In this paper, we give a Hintikkan game semantics for logics of non-sense and prove its correctness. We also discuss how a known solution method in game theory, the iterated elimination of strictly dominated strategies, relates to semantic games for logics of nonsense. Finally, we extend the logics of nonsense only by means of semantic games, developing a new logic of nonsense, and propose a new game semantics for Priest's Logic of Paradox.

~

\end{abstract}

\section{Introduction}

Logics of nonsense allow a third truth value to express propositions that are \emph{nonsense}. The initial motivation behind introducing nonsensical propositions was to capture the logical behavior of semantic paradoxes that were thought to be nonsensical.

As argued by Ferguson, nonsense is ``infectious" \cite{fer0}: ``Meaninglessness [...] propagates through the language; that a subformula is meaningless entails that any complex formula in which it finds itself is likewise meaningless." Logics of nonsense, for this reason, are intriguing.

Nonsense appears in various contexts. In mathematics, the expression ``$1/x$'' is considered meaningless when $x=0$. Certain approaches to truth disregard paradoxes of self-reference and exclude them as \emph{meaningless}. In philosophy of science, some theories are thought to be ``monsters" and are thereby excluded from research programs \cite{lak}. In computer science, software bugs often remain throughout the program run and produce an unexpected output. In all these discussions, nonsense propagates throughout the system and ``infects'' the other truth values.

In what follows we will employ game theoretical tools and techniques to give a broader reading of the aforementioned \emph{infectiousness}. What motivates our approach is manifold. First, game theoretical semantic tools offer a very intuitive and natural approach to semantics.  Moreover, they suggest \emph{computational connections} between truth, proofs, programs and strategies, relating major concepts of game theory, computer science and logic to each other constructively. In this way, certain connections between how errors are handled in programs and how they propagate throughout the programs once they appear can be established. Studying logics of nonsense helps developing a substantial analysis of such situations. Second, game semantics is perhaps the most studied \emph{non-compositional} semantics. It is non-compositional in the way that the truth of a complex formula is evaluated based on the truth values of \emph{some} of its components. We can thus answer how ``some'' truth values can infect the others and computationally determine the truth value of the formula in question. Such ``infections" suggest that there is a certain \emph{strategic} interaction amongst the players. Third, a study of ``infectiousness'' helps us draw a broader picture of \emph{interactive} and rational behaviour, which is a central theme in multi-agent systems, social choice and decision theories. In conclusion, a game theoretical semantic study of logics of nonsense helps us understand ``infectiousness'' from a variety of different but complementary perspectives.

In game semantics, the truth of a formula in a model is evaluated by playing a game. Given a model and a well-formed formula, players try to win a game in order to decide the truth value of the formula in the given model. In classical propositional logic, \emph{the semantic verification game} (or ``semantic game'' for short) is played by two players, the \emph{verifier} and the \emph{falsifier}, who are often called Heloise and Abelard respectively. The verifier aims at verifying the truth of a given formula in a given model whereas the falsifier aims at falsifying it. The game is for two-players, reflecting the binary truth values of the logic.

In a semantic game, the given formula is broken into subformulas step by step. The game terminates when it reaches the propositional atoms. If the game ends with a true atom, then the verifier wins the game. Otherwise, it is a win for the falsifier. The moves and turns of the game are determined syntactically based on the shape of the formula. If the main connective is a conjunction, the falsifier makes a move. If it is disjunction, the verifier makes a move. If the main connective is a negation, the players switch roles: the verifier becomes the falsifier, the falsifier becomes the verifier. Throughout the game, players may switch their roles according to the rules as the rules are given for players' roles not for the individual players.

We say that a player has a \emph{winning strategy} if he has a set of rules that guides him throughout the play and tells him which move to make, and consequently gives him a win regardless of how the opponent plays. In \emph{classical} game semantics, winning strategies necessarily determine the truth values of the formulas. In non-classical game semantics, this assumption is rejected. Because some games may have multiple winners -- with multiple ``winning'' strategies. In that case, we need to be able to identify the winning strategy that necessarily determines the truth value of the formula in question. We call such winning strategies \emph{dominant}. We will use dominant strategies to capture the infectiousness of nonsense.

Logics of nonsense and game semantics both have a long, but relatively unconnected, history.

Game semantics has been a popular research area since Jaakko Hintikka, and Helsinki School researchers produced a significant amount of work on the subject. Broad historical surveys of the subject were presented in \cite{piet,hin8} with many references. Game semantics have been used to understand proofs and computation in computer science \cite{abr2}, and dialogical games in philosophical logic \cite{lor1,rah0}.

The connection between non-classical logic and game semantics have \emph{not} been studied so extensively. Game semantics for various many-valued logics have been presented in \cite{ferm,ferm0}. In the context of dialogical logic, which can be seen as a variant of game semantics, the connection between semantic games and paraconsistency was explored in \cite{rah}. A recent work focused on the connection between various well-known paraconsistent logics and game semantics and offered various modifications for the games based on different logics \cite{bas25,bas30}. The current work attempts at following the same path yet puts more emphasis on strategies.

Logics of nonsense were first suggested by Bochvar in the 1930s and by Halld\'{e}n in the 1940s \cite{boc,hald}\footnote{The translation of \cite{boc} appeared as \cite{boc0}.}. These logics enjoy different validities as Halld\'{e}n's system preserves nonsense (along with true) under valid inferences. Bochvar's system, however, does not. As such, Halld\'{e}n's logic is paraconsistent; Bochvar's, dually, is paracomplete. Such systems were extended by {\AA}qvist and Segerberg \cite{aqv,seg1}. Ha\l{}kowska suggested an algebra for logics of nonsense \cite{halk}. In that work, nonsense truth value was motivated by certain meaningless mathematical expressions such as  $1/x$ where $x = 0$.

There has been an increasing interest towards infectious logics. Ferguson examined the logics of nonsense in relation to various other non-classical containment logics, offering a broad perspective \cite{fer0}. Szmuc discussed logics of nonsense within the context of some other non-classical logics, including logics of formal inconsistency and of formal underminedness. As such, he carried the issue into broader classes of logics and generated a variety of logics of nonsense \cite{szm}. Following this methodology, Ciuni et al constructed a linear order of infectious logics, extending logics of non-sense into a countably infinite family of subsystems of classical logic \cite{ciu}. What they achieved in their work uses logical and semantical methods, while we arrive at similar systems using game theoretical tools in this work. Omori presented an extension of logics of nonsense by discussing them within the framework of logics of formal inconsistency \cite{omo1}. Omori and Szmuc focused on the behavior of conjunction and disjunction in logics of nonsense and observe the unexpected working of disjunction \cite{omo0}.

The current paper is organized as follows. First, we formally define semantic games. Then, after a brief discussion of a logic of nonsense, we introduce a game semantics for it. Consequently, we prove the correctness theorem of our game semantics. We then extend our semantic games for an attempt to obtain some other logics of nonsense that extend the initial system. Finally, as a case study, we take advantage of our approach in order to obtain a more refined semantic game for a well-known paraconsistent logic, Priest's Logic of Paradox.

The main contribution of the paper is to emphasize the role of dominant strategies in game semantics in the context of non-classical logics. By doing so, we establish a stronger connection between non-classical logics and game theory. We achieve this by showing how dominant strategies and certain truth values in game semantics work similarly.

This paper is part of a research programme on \emph{non-classical game theory}. Similar to non-classical approaches to truth (and proof), it is possible to analyze strategies and wins from a non-classical perspective. This will provide a more nuanced understanding of game theory and open up a new avenue for applications for non-classical logic(s). 

 Throughout this work, we will use ``game theoretical semantics'' and ``game semantics'' interchangeably, assuming no confusion arises.

\section{Semantic Games}

Non-classical logics and their game semantics allow us to study the connection between non-classical logical elements and their corresponding game theoretical properties. How does a paraconsistent game look like? Can we have a multi-player semantic game for a multi-valued logic? Whilst answering these questions, one can \emph{engineer non-classical logics} using game theoretical tools as well as \emph{develop semantic games} for non-classical logics. Those games can be multi-player, non-zero-sum, cooperative, non-determined and non-sequential, differing from the semantic games for classical logic. 

In what follows, we focus on systems of propositional logic which allow us to reason about strategies in semantic games. This will exemplify the fact that strategies in classical semantic games are over-simplified (or \emph{degenerate}) forms of strategies that appear more clearly in semantic games for some non-classical logics, such as logics of nonsense. Before presenting our results, we need to set up our system.

We define semantic games following the terminology given in \cite{piet,bas25}. First, we consider the language $\mathcal{L}$ of propositional logic given as follows in the Backus--Naur form for a set of countable propositional variables $\mathbf{P}$:
$$\varphi:=~ p ~\mid~ \neg \varphi ~\mid~ \varphi \wedge \varphi ~\mid~ \varphi \vee \varphi$$
where $p \in \mathbf{P}$. We define the conditional arrow as expected: $\varphi \rightarrow \psi \equiv \neg \varphi \vee \psi$.

A model for semantic games is defined as follows.

\begin{dfn} A model $M$ is a tuple $(S, v)$ where $S$ is a non-empty domain on which the game is played, and the valuation function $v$ assigns the formulas in $\mathcal{L}$ to truth values in the logic.	
\end{dfn}

In semantic games, we have a set of players, game rules and a set of positions. Some non-classical logics, such as Priest's LP and Strong Kleene Logic, share the same truth tables but they differ on their set of designated truth values -- rendering the former paraconsistent, the latter paracomplete. In order to distinguish such systems, designated truth values\footnote{The set of designated truth values are used to define theorems in a particular logic. They can be viewed as non-classical extensions of the truth value \textsc{True}, and are preserved under valid inferences. The relation between the proof theory and semantic games is well-pronounced for certain logics, which falls outside the scope of this paper \cite{abr2}. Similarly, game theoretical readings of logical consequence relation remains a curious problem. We are thankful to the anonymous referee for pointing this out.} are specified -- even though we will not focus on them in this work. Some games may not be sequential, thus a \emph{game-token} may be used to indicate the current position of the players.

\begin{dfn} A semantic game is a tuple $\Gamma = (\pi, \rho, \sigma, \tau, \delta)$ where $\pi$ is the set of players, aiming at winning the game by reaching atomic formulas with specific truth values based on their roles, $\rho$ is the set of well-defined game rules, $\sigma$ is the set of positions, $\tau$ is the set of positions of the game-token in the case of a concurrent play, and $\delta$ is the set of designated truth values. 
\end{dfn}

Let us start by explaining how $\Gamma$ is constructed. Players in the semantic game $\Gamma$ have interchangeable roles. A player $p_i \in \pi$ may assume different roles, thus strategies, in the game. Therefore, the game rules will be given not for the players $p_i$ but for their \emph{roles}. Players, then, will follow the rules based on what roles they assume. The \emph{falsifier} aims at reaching a false atom, the \emph{verifier} a true atom.

The positions $\sigma$ in the game are determined by the subformulas of the given formula and the players. As such, they also include the turn function, which specifies which players are allowed to make a move at each position. Consequently, the set of positions will be composed of tuples as $(p_i, \varphi)$ for $p_i \in \pi$ and a well defined formula $\varphi$. The tuple $(p_i, \varphi)$ will read ``it is player $p_i$'s turn at $\varphi$". The set $\sigma_{p_i}$ will denote the set of positions for player $p_i \in \pi$, and will be defined as $ \sigma_{p_i} = \{(p_i, \varphi) : (p_i, \varphi) \in \sigma \text { for a fixed player } p_i\}$. For simplicity, we will include the empty set in $\sigma$.

The set $\sigma$ is not sufficient by itself to describe which positions are concurrent and played simultaneously. In semantic games for classical logic, $\tau$ is the set of singletons as there is no concurrent or parallel play. In games, where two players are allowed to make moves at the same time, we will have $\{(p_1, \varphi), (p_2, \psi)\} \in \tau$, which reads ``player $p_1$ plays at $\varphi$ and player $p_2$ plays at $\psi$ simultaneously".

Designated truth values $\delta$ determine the theorems in a logic. In order to distinguish different logics, we specify $\delta$. Yet, we will not focus on proof theory in this work.

The rules $\rho$ of semantic games are defined inductively as (partial) transformations from a game position $(p_i, \varphi)$ to a set of game positions $\{(p_j, \psi)\}_{j \in I}$ for $p_i, p_j \in \pi$, $I \subseteq \pi$, well-defined formula $\varphi$ and a subformula $\psi$ of $\varphi$, defined in the standard way. As such, game rules do not change if a position is a concurrent or not since they are defined per position. For simplicity reasons, we do not include a turn function to determine which players are supposed to make a move at any position. Game rules and positions specify the turn. If player $p_i$ is not supposed to make a move at position $(p_i, \varphi)$ then, the game rule $\rho_j$ will return $\rho_j((p_i, \varphi)) = \emptyset$ where $\emptyset \in \sigma$. 

A \emph{run} (or \emph{play}) is a sequence of sets from $\tau$ which starts with $\{ \{ p_i, \varphi \} \}$ and ends with a set that contains a position with an atom,  where $\varphi$ is the initially given formula. In semantic games for classical propositional logic, for example, a run is a sequence of singletons from $\tau$. 

The semantic game $\Gamma$ for a model $M$ and a well-defined formula $\varphi \in \mathcal{L}$ is denoted by $\Gamma(M, \varphi)$.

In semantic games, logical behavior is described using game theoretical tools and concepts. Strategies are perhaps the most important of such tools. In semantic games for classical logic, having winning strategies \emph{necessarily} determines the truth value of the formula in question, and \emph{vice versa}. This is a very intuitive and fundamental idea. Yet, it fails to translate into various non-classical logics as having a winning strategy for a semantic game does not necessarily determine the truth value of the formula in question. For that reason, we need a stronger concept.

\begin{dfn}{\label{dfn-dominant-strategy}}
In $\Gamma(M, \varphi)$, a \emph{dominant winning strategy} is a winning strategy that determines the truth value of $\varphi$ once played.
\end{dfn}

In other words, in $\Gamma(M, \varphi)$ if a player $p_i$ admits a dominant winning strategy $\Sigma$ and if $p_i$ plays $\Sigma$, then $\varphi$ has the truth value that $p_i$ forces. If $p_i$ admits a winning strategy, on the other hand, then $\varphi$ \emph{may} have the truth value that $p_i$ forces, depending on other players and  conditions. However, if $p_i$ does not admit a winning strategy nor a dominant winning strategy, she cannot determine the truth value of $\varphi$. Naturally, in classical logic, every winning strategy is dominant. In non-classical logics, they are not.

Notice that what Definition~\ref{dfn-dominant-strategy} describes is rather different than what most readers may be familiar in classical game theory about strategy dominance.\footnote{We are thankful to the anonymous referee for pointing this out.} The standard game theoretical definitions discuss pay-offs and higher pay-offs, whereas we focus on determining truth-values \cite{osb}. There is, however, a philosophical connection. Higher-payoffs with dominant strategies help a player win a game in competitive games, dominant winning strategies help a player to determine the truth value of a formula in semantic games.

\section{Game Semantics for the Bochvar--Halld\'{e}n Logic}

Bochvar--Halld\'{e}n Logic introduces an additional truth value $N$, called \emph{nonsense}, which intuitively stands for sentences which are nonsensical or meaningless. As we argued earlier, Bochvar--Halld\'{e}n logics are actually two distinct logics with the same truth table. In Bochvar's system, the designated truth value is $T$ whereas in Halld\'{e}n's it is $\{ T, N\}$. For our semantic considerations, we treat them together and call this formalism the Bochvar--Halld\'{e}n Logic (BH3, for short) with the following truth table.

\begin{figure}[h!]
\centering
\begin{tabular}{l|l}
& $\neg$ \\
\hline
$T$ & $F$ \\
$N$ & $N$ \\
$F$ & $T$ 
\end{tabular} 
\qquad
\begin{tabular}{c|c|c|c}
$\wedge$ & $T$ & $N$ & $F$ \\
\hline
$T$ & $T$ & $N$ & $F$ \\
$N$ & $N$ & $N$ & $N$ \\
$F$ & $F$ & $N$ & $F$
\end{tabular}
\qquad
\begin{tabular}{c|c|c|c}
$\vee$ & $T$ & $N$ & $F$ \\
\hline
$T$ & $T$ & $N$ & $T$ \\
$N$ & $N$ & $N$ & $N$ \\
$F$ & $T$ & $N$ & $F$
\end{tabular}
\caption{\emph{The truth tables underlying} BH3.}
{\label{truthtable}}
\end{figure}

In a semantic game for BH3, we need three players to force the three truth values. In addition to the classical players Verifier and Falsifier, we introduce a third player which we call ``Dominator''. Dominator forces the game to a nonsense proposition. As such, he is allowed to make moves along other players. We also stipulate that Dominator's strategy is \emph{dominant} -- his wins determine the truth value.

We denote the semantic game for BH3 by GTS\textsuperscript{BH3} and propose the following rules for GTS\textsuperscript{BH3}.

\begin{dfn}{\label{dfn-BH3game}} The tuple $\Gamma_{\mathrm{BH3}} = (\pi, \rho, \sigma,\tau,\delta)$ is a semantic game for BH3 where $\pi=\{\text{Falsifier, Verifier},$ $\text{Dominator}\}$, $\sigma$ is the set of tuples $(p_i, \varphi)$ for $p_i \in \pi$ and a well-formed formula $\varphi$, and $\delta$ is $\{ T, N \}$ for Bochvar logic and $\{ T \}$ for Halld\'{e}n logic. For a game $\Gamma_{\mathrm{BH3}}(M, \varphi)$, the set of positions $\sigma$ is specified in the following.
\begin{itemize}
	\item If $\varphi$ is atomic, then $(p_i, \varphi) \in \sigma$ for all $p_i \in \pi$,
	\item If $\varphi = \neg \psi$, then $(p_i, \varphi) \in \sigma$ for all $p_i \in \pi$, and $(p_j, \psi) \in \sigma$ for some $p_j \in \pi$ depending on $\psi$'s main connective,
	\item If $\varphi = \chi \wedge \psi$, $(\mathrm{Falsifier}, \varphi) \in \sigma$, $(\mathrm{Dominator}, \varphi) \in \sigma$, and $(p_j, \chi), (p_k, \psi) \in \sigma$ for some $p_j, p_k \in \pi$ depending on $\psi$ and $\chi$'s main connectives,
	\item If $\varphi = \chi \vee \psi$, $(\mathrm{Verifier}, \varphi) \in \sigma$, $(\mathrm{Dominator}, \varphi) \in \sigma$, and $(p_j, \chi), (p_k, \psi) \in \sigma$ for some $p_j, p_k \in \pi$ depending on $\psi$ and $\chi$'s main connectives.
\end{itemize}

The set $\tau$ is given inductively as follows only for those positions $(p_i, \varphi)$ in $\sigma$ that are allowed to be played simultaneously but independently.
\begin{itemize}
\item For $\varphi = \chi \wedge \psi$, $\{(\mathrm{Falsifier}, \chi \wedge \psi), (\mathrm{Dominator}, \chi \wedge \psi)\} \in \tau$,
\item For $\varphi = \chi \vee \psi$, $\{(\mathrm{Verifier}, \chi \vee \psi), (\mathrm{Dominator}, \chi \vee \psi)\} \in \tau$.
\end{itemize}

Finally, $\rho$ is given inductively as follows for players' roles.
\begin{itemize}
\item[($\rho_p$)] If $\varphi$ is atomic, the game terminates, and Verifier wins if $\varphi$ is true, Falsifier wins if $\varphi$ is false and Dominator wins if $\varphi$ is nonsense;
\item[($\rho_\neg$)] If $\varphi = \neg \psi$, Falsifier and Verifier switch roles, Dominator keeps his role, and the game continues as $\Gamma_{\mathrm{BH3}}(M, \psi)$;
\item[($\rho_\wedge$)] If $\varphi = \chi \wedge \psi$, Falsifier and Dominator choose between $\chi$ and $\psi$ independently and simultaneously;
\item[($\rho_\vee$)] If $\varphi = \chi \vee \psi$, Verifier and Dominator choose between $\chi$ and $\psi$ independently and simultaneously;
\item[($\rho_s$)] Dominator's strategy strictly dominates the Verifier's and Falsifier's. 
\end{itemize}
\end{dfn}

Let us now briefly explain the game. First, the set of positions $\sigma$ is constructed inductively. Atomic formulas are associated with all players as they all need to check whether they win. At negated formulas, players reshuffle their roles and the game continues with the players based on the next formula's syntactic form, which is determined inductively. For conjunctions and disjunctions, the next position is determined based on the main connective of conjuncts and disjuncts, respectively. Two players are allowed to make moves at the same time but independently at conjunctions and disjunction, as specified by $\tau$: conjunctions are for Falsifier and Dominator, disjunctions are for Verifier and Dominator. The set of available moves for those players are specified further in the game rules $\rho$. 

We stipulate that Dominator's strategy dominates the others and his role does not change throughout the game, even under negation, which reflects the truth table for BH3 given in Figure~\ref{truthtable}. The strategies of Verifier and Falsifier, however, do not dominate each other or any other strategy, by default.

The following example illustrates how BH3 semantic games are structured and played.

\begin{ex}{\label{example-one}}
Consider the formula $(p \vee q) \vee (r \wedge q)$ where $p$ is true, $q$ is nonsense and $r$ is false. This formula has the truth value $N$ in BH3. The following diagram depicts the game tree informally.

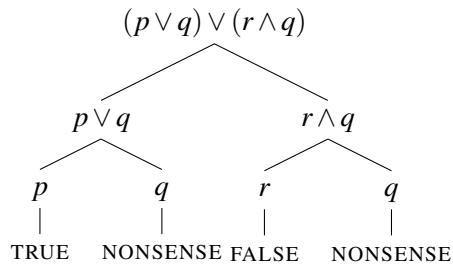
\begin{quote}
	\Tree [.{$(p \vee q) \vee (r \wedge q)$} [.{$p \vee q$} [.$p$ \small{\textsc{true}} ] [.$q$ \small{\textsc{nonsense}} ] ] !\qsetw{3cm} [.{$r \wedge q$} [.$r$ \small{\textsc{false}} ] [.$q$ \small{\textsc{nonsense}} ] ] ]

\end{quote}

In this game, some of the positions in $\sigma$ are $(\mathrm{Verifier}, (p \vee q) \vee (r \wedge q))$, $(\mathrm{Dominator}, (p \vee q) \vee (r \wedge q))$, $(\mathrm{Verifier}, p \vee q)$, $(\mathrm{Dominator}, p \vee q)$, $(\mathrm{Falsifier}, r \wedge q)$, $(\mathrm{Dominator}, r \wedge q)$ (and those for the atoms). The token set $\tau$ include sets $\{ (\mathrm{Verifier}, (p \vee q) \vee (r \wedge q)), (\mathrm{Dominator}, (p \vee q) \vee (r \wedge q)) \}$ and $\{ (\mathrm{Verifier}, p \vee q), (\mathrm{Dominator}, p \vee q) \}$, for example.

Now, at the beginning of the game, Verifier and Dominator make choices. If Verifier chooses $p \vee q$, then he gets to make the next move and chooses $p$ to win the game. So, he has a \emph{winning} strategy. However, Verifier's winning strategy is dominated. At the beginning, Dominator also gets to make a move. Suppose, he chooses $p \vee q$ as well. Now, he can make a move again and chooses $q$ which is nonsense. This constitutes his winning strategy. But, by game rules, his strategy dominates the others. So, he has a dominant winning strategy, which determines the truth value of the formula $(p \vee q) \vee (r \wedge q)$. The reasoning for Dominator's choice of $r \wedge q$ is similar.

The play/run for this specific example is given as follows:

$\lbrace ~ \{ (\mathrm{Verifier}, (p \vee q) \vee (r \wedge q)), (\mathrm{Dominator}, (p \vee q) \vee (r \wedge q)) \}$,

$\{ (\mathrm{Verifier}, p \vee q), (\mathrm{Dominator}, p \vee q) \}$,

$\{ (\mathrm{Verifier}, p ), (\mathrm{Dominator}, q) \} ~ \rbrace$.
\hfill $\blacktriangle$
\end{ex}

Let us now start with noting that Dominator is the only dominant player.

\begin{thm}{\label{theo-dom}}
In a GTS\textsuperscript{BH3} semantic game $\Gamma_{\mathrm{BH3}}(M, \varphi)$, Verifier and Falsifier can never have winning strategies at the same time.
\end{thm}

\begin{proof}
The proof is by induction on $\varphi$.

For the propositional case, by definition, either of the players will have a winning strategy. Similarly, for the negation, only one player gets to make a move in each case.

The interesting cases are the binary boolean connectives. If $\varphi = \psi \wedge \chi$, then, by the game rules, Falsifier and Dominator gets to make moves. If they have winning strategies at that stage for some of the conjuncts, they will still have winning strategies at $\psi \wedge \chi$. On the other hand, Verifier can have a winning strategy at a conjunction as well. The only condition for this case is that Verifier must have a winning strategy for both of the conjuncts $\psi$ or $\chi$. In this case, neither Falsifier nor Dominator has a winning strategy. At the end, we have two cases: (i) Falsifier and/or Dominator has a winning strategy, or (ii) Verifier has a winning strategy. Either case, Verifier and Falsifier cannot have winning strategies at the same time.

The case for the disjunction is similar.
\end{proof}

\begin{cor}{\label{cor-det}}
BH3 semantic games GTS\textsuperscript{BH3} are determined. 	
\end{cor}

Dominant winning strategies and game rules allow us to establish the following correctness theorem for the game semantics for BH3.

\begin{thm}{\label{theo-correct}} 
In a GTS\textsuperscript{BH3} semantic game $\Gamma_{\mathrm{BH3}}(M, \varphi)$,

\begin{itemize}
\item Verifier has a dominant winning strategy if and only if $\varphi$ is true in $M$,
\item Falsifier has a dominant winning strategy if and only if $\varphi$ is false in $M$,
\item Dominator has a dominant winning strategy if and only if $\varphi$ is nonsense in $M$.
\end{itemize}\end{thm}

\begin{proof}
We start with the case for Verifier. Let us consider the game state $(\mathrm{Verifier}, \varphi)$. We proceed by induction on the complexity of $\varphi$.

\textsc{Propositional Variables for Verifier:} If $\varphi$ is a propositional letter $p$ true in $M$, then Verifier wins by definition, hence has a dominant winning strategy. By Theorem~\ref{theo-dom}, Falsifier cannot have a winning strategy for $p$. Additionally, by definition, Dominator does not have a winning strategy for $p$. Thus, Verifier's strategy is dominant.

Conversely, if Verifier has a dominant strategy for the game for $p$, by definition, $p$ is true in $M$.

\textsc{Negation for Verifier:} Let $\varphi = \neg \psi$ be true. Then, by the truth table $\psi$ is false. By the game rules, the play continues where Verifier becomes Falsifier. By the induction hypothesis (for falsifier), Falsifier has a dominant winning strategy for $\psi$. Then, Verifier has a winning strategy for $\neg \psi$ by simply playing her game as Falsifier for $\psi$. Her strategy at $\neg \psi$ is dominant as no additional move by the other players is introduced at this stage of the game, according to the game rules. Thus, Verifier has a dominant winning strategy for $\varphi$.

Conversely, assume that Verifier has a dominant winning strategy for $\varphi = \neg \psi$. Then, the game carries on where Falsifier has a dominant strategy for $\psi$. By the induction hypothesis, then $\psi$ is false. By the truth table, $\neg \psi$ is true, rendering $\varphi$ true.

\textsc{Conjunction for Verifier:} Now, let $\varphi$ be a conjunction of the form $\chi \wedge \psi$. Assume that $\varphi$ is true. According to the truth table, the only way to make it true is to have both $\chi$ and $\psi$ true. Then, by the induction hypothesis, Verifier has dominant winning strategies for both $\chi$ and $\psi$. However, for $\varphi$, Falsifier and Dominator make moves. But, whichever move they make (whichever of $\chi$ or $\psi$ they choose), it will be a win for Verifier. Thus, Verifier has a winning strategy. Is her strategy dominant? Notice that neither Falsifier nor Dominator admits a winning strategy at this stage, thus Verifier's strategy is  dominant. Thus, for $\varphi$, Verifier has a dominant winning strategy.

Conversely, assume that Verifier has a dominant winning strategy for $\varphi$ which is of the form $\chi \wedge \psi$. At this stage, it is Falsifier and Dominator who make moves. By Theorem~\ref{theo-dom}, we know that Falsifier cannot have a dominant winning strategy. Thus, he cannot have a winning strategy. By assumption, we know that Dominator does not have a winning strategy either. If he did, it would have dominated Verifier's and we couldn't have assumed that Verifier has a dominant winning strategy for $\varphi$. Now, whatever choice Falsifier and Dominator make, Verifier still has a dominant strategy for the subgames for $\chi$ and $\psi$. Thus, by the induction hypothesis $\chi$ and $\psi$ are both true. By the truth table, their conjunction is true, rendering that $\varphi$ is true.

\textsc{Disjunction for Verifier:}
Let $\varphi$ be a disjunction of the form $\chi \vee \psi$. Assume that $\varphi$ is true. Then, one of the disjuncts is true, and neither is nonsense. As such, by the induction hypothesis, Dominator does not have a winning strategy for either of the disjuncts $\chi$ or $\psi$, and consequently does not have a dominant winning strategy. Verifier then makes a move for either $\chi$ or $\psi$, whichever is true. Then, choosing the true disjunct is her winning strategy at $\varphi$, independent from whatever Dominator chooses. Verifier's strategy is also dominant as Dominator does not have a winning strategy.

Conversely, let Verifier have a dominant winning strategy for $\varphi = \chi \vee \psi$. Then, Dominator does not have a winning strategy at this stage, which, together with the induction hypothesis, suggest that neither of the disjunct is nonsense. Now, Verifier makes a choice following her dominant winning strategy, say $\chi$, without loss of generality. By the induction hypothesis, then, $\chi$ is true. We also established that the other disjunct cannot be nonsense. Thus, by the truth table, $\varphi$ is true. 

\textsc{The Cases for Falsifier:} The cases for Falsifier are similar to those of Verifier's, hence skipped.

\textsc{Propositional Variables for Dominator:} Let $\varphi$ be a propositional variable $p$. If $p$ is nonsense in $M$, then by definition Dominator has a dominant winning strategy.

Conversely, let Dominator have a dominant winning strategy for a propositional variable $p$. Then, by definition $p$ is nonsense.

\textsc{Negation for Dominator:} Let $\varphi = \neg \psi$. If $\varphi$ is nonsense in $M$, then by the truth table $\psi$ is nonsense, too. By the induction hypothesis, then Dominator has a dominant winning strategy for $\psi$. He then continues with the same strategy for $\neg \psi$. His strategy remains dominant as no other player is allowed to make a move at this stage along with him. Even if they did, his strategy would dominate theirs. Thus, Dominator has a dominant winning strategy for $\varphi$.

Conversely, let Dominator have a dominant winning strategy for $\varphi = \neg \psi$. Then, he carries on with his strategy and his role for $\psi$, by the game rules. Then, Dominator has a dominant winning strategy for $\psi$. Then, by the induction hypothesis $\psi$ is nonsense. By the truth table, the negation of nonsense is still nonsense, rendering $\varphi$ nonsense as well.

\textsc{Conjunction for Dominator:} Let $\varphi = \psi \wedge \psi$. If $\psi \wedge \psi$ is non-sense, then by the truth table at least one of the conjuncts has the truth value $N$. Dominator is allowed to make a move at a conjunction and he chooses a conjunct with the truth value $N$. Thus, he has a winning strategy. But, is his strategy a dominant one? At a conjunction, Falsifier can make a move, too. It is also possible that for a nonsense conjunction  $\psi \wedge \psi$, one of the conjunct can be false, according to the truth table. Thus, Falsifier can very well choose that conjunct. In that case, he can have a winning strategy as well but his strategy is not a dominant winning strategy as Dominator's strategies dominate those of Falsifier's (and verifier's) by the game rules. Hence, Dominator has a dominant winning strategy at a nonsense conjunction.

Conversely, again, let Dominator have a dominant winning strategy for $\varphi = \psi \wedge \psi$. Dominator is allowed to make a choice at conjunctions, and he chooses a conjunct following his winning strategy - say he chooses $\psi$, without loss of generality. His dominant winning strategy perseveres after selecting $\psi$. According to the game rules, no other player can dominate his stratagies. By the induction hypothesis, then, $\psi$ is nonsense. According to the truth table, conjunction of nonsense with anything produces nonsense. Thus, $\varphi$ is nonsense.

\textsc{Disjunction for Dominator:} Let $\varphi = \psi \vee \psi$. If $\psi \vee \psi$ is non-sense, then, similarly, by the truth table at least one of the disjuncts has the truth value $N$. Since Dominator is allowed to make a move at disjunctions, he chooses the nonsense disjunct. By the game rules, again, his winning strategy is a dominant one.

Conversely, let Dominator have a dominant winning strategy for $\varphi = \psi \vee \psi$. Similarly, he chooses the nonsense disjunct, which renders the whole disjunction $\varphi$  nonsense.

This completes the proof.	
\end{proof}

It is important to note that in Theorem~\ref{theo-correct}, Verifier or Falsifier can only have a dominant strategy if Dominator does not have one.

\subsection*{More on Strategies}

Dominant strategies can be useful in \emph{solving} games. A familiar method, called the \emph{iterated elimination of strictly dominated strategies} (IESDS, for short), suggests that by eliminating those strategies that are dominated, we can reach a solution. This method directly applies to semantic games for BH3. The reasons can be mentioned as follows.

\begin{obs} In semantic games BH3, the strategies for the dominant player strictly dominates the strategies of Verifier and Falsifier.	
\end{obs}

\begin{obs}
In a GTS\textsuperscript{BH3} semantic game, Dominator makes a move at each connective.
\end{obs}

\begin{rem} \cite{osb} An action of a player in a finite strategic game is never a best-response if and only if it is strictly dominated.
\end{rem}

Intuitively, IESDS applies to semantic games for BH3 because of the \emph{infectiousness} of the truth value nonsense. If the given formula contains an atom with the truth value nonsense, all the other paths in the game tree can be eliminated as those strategies are strictly dominated even if they are winning (non-classically speaking). The following example illustrates our point.  

\begin{ex}
Consider the formula $(p \vee q) \vee (r \wedge q)$ from Example~\ref{example-one} where $p$ is true, $q$ is nonsense and $r$ is false. This formula has the truth value $N$ in BH3.
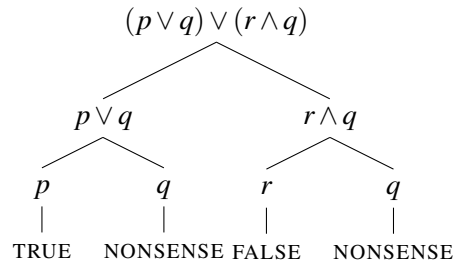
\begin{quote}
	\Tree [.{$(p \vee q) \vee (r \wedge q)$} [.{$p \vee q$} [.$p$ \small{\textsc{true}} ] [.$q$ \small{\textsc{nonsense}} ] ] !\qsetw{3cm} [.{$r \wedge q$} [.$r$ \small{\textsc{false}} ] [.$q$ \small{\textsc{nonsense}} ] ] ]
\end{quote}
In this case, Verifier's strategy of $L-L$ (which means play ``left" consecutively twice) is dominated by Dominator's strategy $L-R$. The strategy $L-L$ is eliminated, yielding $L-R$ as the dominant strategy. Similarly, Falsifier's strategy if $R$ is played by Verifier at the beginning of the game is eliminated by Dominator's strategy of $R-R$. \hfill $\blacktriangle$
\end{ex}

In game theory, the strategy obtained after IESDS is often called a \emph{rationalizable} strategy \cite{osb}. Similarly, we call the dominant strategy in a semantic game obtained after IESDS as the \emph{truth-maker strategy} since such strategies determine the truth value of the formula in question. Truth-maker strategies exist in many semantic games.

\begin{prop}
Every GTS\textsuperscript{BH3} semantic game has a truth-maker strategy.
\end{prop}

\begin{proof}
Follows from Theorem~\ref{theo-dom} and Corollary~\ref{cor-det}.
\end{proof}

The following theorem summarizes our observations.

\begin{thm}{\label{theo-elimination}}
In a GTS\textsuperscript{BH3} semantic game $\Gamma_{\mathrm{BH3}}(M, \varphi)$, if $\varphi$ contains a literal with a truth value nonsense, then Dominator has a dominant winning strategy and consequently $\varphi$ is nonsense.
\end{thm}

\begin{proof}
It is easy to prove this theorem by an induction on $\varphi$ using logical techniques, reflecting the truth table for BH3 and Theorem~\ref{theo-correct}. However, in what follows, we will prove it using the method of IESDS.

The cases for the propositional variables and the negation are immediate, hence skipped.

Assume that $\varphi$ is a conjunction of the form $\psi \wedge \chi$. Without loss of generality let us suppose that $\psi$ contains a literal with a truth value nonsense. By the induction hypothesis, Dominator has a dominant winning strategy for $\psi$. Thus, Dominator's strategy for $\chi$ is strictly dominated. Then, we can eliminate the strategy for $\chi$. In this case, Dominator chooses $\psi$ when it is his turn at $\varphi$. Similarly, by the game rule $(\rho_s)$, any possible (winning) strategy of Falsifier is strictly dominated, thus can be eliminated. Hence, Dominator's strategy remains dominant at $\varphi$ and determines the outcome of the game where $\varphi$ contains a literal with a truth value of nonsense. Hence, by Theorem~\ref{theo-correct}, $\varphi$ is nonsense.

The case for disjunction is similar.
\end{proof}

An interesting question is how such solution methods relate to Nash equilibria in semantic games. There can be thought of at least two approaches to this question. First, one can  \emph{engineer} logics that necessarily have an equilibria which can be obtained by IESDS. In other words, one can try to construct a truth table and a logic where the players' strategies can be worked out using IESDS. What kind of logics can be produced that way? What is their algebraic structure? Second, one can try to \emph{develop} a game with a rationalizable strategy for a logic which is not infectious. These are big questions. They fall outside the scope of this paper and are left for future work.

Instead of eliminating strictly dominated strategies, one can introduce additional players with strictly dominant strategies into semantic games for BH3. What kind of a logic does this methodology generate? We answer this question in the next section.

\section{Extending Bochvar--Halld\'{e}n Games}

Using the idea of dominant strategies, we can \emph{engineer} some semantic games without first considering their possible semantics. This is an interesting methodology which develops logics (and truth tables) which are \emph{solely} generated by semantic games.

In what follows, we extend the BH3 games to four players, where we call the forth player \emph{Dictator}. Dominator's strategy dominates Falsifier's and Verifier's, whereas Dictator's strategy dominates them all. For simplicity, we call the truth value that is forced by Dictator as \emph{super} and denote it by $S$. Thus, Dictator's role is to force the semantic game to an end with an atom with the truth value super.

We propose the following rules for the extended game $\Gamma(M, \varphi)$ with respect to players' roles.

\begin{itemize}
\item[($\rho'_p$)] If $\varphi$ is atomic, the game terminates, and Verifier wins if $\varphi$ is true, Falsifier wins if $\varphi$ is false, Dominator wins if $\varphi$ is nonsense, and Dictator wins if $\varphi$ is super,
\item[($\rho'_\neg$)] if $\varphi = \neg \psi$, Falsifier and Verifier switch roles, Dominator and Dictator keep their roles, and the game continues as $\Gamma(M, \psi)$,
\item[($\rho'_\wedge$)] if $\varphi = \chi \wedge \psi$, Falsifier, Dominator and Dictator choose between $\chi$ and $\psi$ simultaneously,
\item[($\rho'_\vee$)] if $\varphi = \chi \vee \psi$, Verifier, Dominator and Dictator choose between $\chi$ and $\psi$ simultaneously.
\item[($\rho'_s$)] Dominator's strategy strictly dominates Verifier's and Falsifier's, and Dictator's strategy strictly dominates them all. 
\end{itemize}

These rules give raise to a logic. We call this system BH4. We can then construct a truth table for BH4 that corresponds to the given game rules in Figure~\ref{BH4truth}.\footnote{It turns out that this logic has been suggested in \cite{szm,ciu}.}

\begin{figure}[h!]
\centering
\begin{tabular}{l|l}
& $\neg$ \\
\hline
$T$ & $F$ \\
$F$ & $T$ \\
$N$ & $N$ \\
$S$ & $S$
\end{tabular} 
\qquad
\begin{tabular}{c|c|c|c|c}
$\wedge$ & $T$ & $N$ & $S$ & $F$ \\
\hline
$T$ & $T$ & $N$ & $S$ & $F$ \\
$N$ & $N$ & $N$ & $S$ & $N$ \\
$S$ & $S$ & $S$ & $S$ & $S$ \\
$F$ & $F$ & $N$ & $S$ & $F$
\end{tabular}
\qquad
\begin{tabular}{c|c|c|c|c}
$\vee$ & $T$ & $N$ & $S$ & $F$ \\
\hline
$T$ & $T$ & $N$ & $S$ & $T$ \\
$N$ & $N$ & $N$ & $S$ & $N$ \\
$S$ & $S$ & $S$ & $S$ & $S$ \\
$F$ & $T$ & $N$ & $S$ & $F$
\end{tabular}
\caption{\emph{The truth tables underlying} BH4.}
{\label{BH4truth}}
\end{figure}

Let us start by making an observation regarding the strategy configuration in the game.

\begin{thm}
In a GTS\textsuperscript{BH4} semantic game $\Gamma_{\mathrm{BH4}}(M, \varphi)$, Verifier and Falsifier can never have winning strategies at the same time. Thus, GTS\textsuperscript{BH4} games are determined.
\end{thm}

\begin{proof}
Similar to that of Theorem~\ref{theo-dom}, hence left to the reader.	
\end{proof}

Similar to what we have shown for BH3, we present the correctness theorem for BH4 as follows.

\begin{thm}{\label{theo-bh4}}
In a GTS\textsuperscript{BH4} semantic game $\Gamma_{\mathrm{BH4}}(M, \varphi)$,

\begin{itemize}
\item Verifier has a dominant winning strategy if and only if $\varphi$ is true in $M$,
\item Falsifier has a dominant winning strategy if and only if $\varphi$ is false in $M$,
\item Dominator has a dominant winning strategy if and only if $\varphi$ is nonsense in $M$,
\item Dictator has a dominant winning strategy if and only if $\varphi$ is super in $M$,
\end{itemize}
\end{thm}

\begin{proof}
The proof is by induction on $\varphi$ and very similar to the proof of Theorem~\ref{theo-correct}. We will only consider some interesting cases.

\textsc{Conjunction for Falsifier:} Let $\varphi = \psi \wedge \chi$. Assume that Falsifier has a dominant winning strategy for $\varphi$. He chooses the false conjunct, say $\psi$. By the induction hypothesis, $\psi$ is false. But, it is not sufficient to establish the truth value of the formula. However, since Falsifier's strategy is the dominant one by assumption, that means that Dominator and Dictator do not have winning strategies. By the induction hypothesis for the cases for Dominator and Dictator in this very theorem, then neither of the conjuncts is nonsense nor super. Thus, according to the truth table of BH4 given in Figure~\ref{BH4truth}, $\varphi$ is false.

Conversely, assume $\varphi$ is false. Falsifier, Dominator and Dictator make moves. According to the truth table, one of the conjunct has to be false, say $\psi$, and Falsifier chooses it. By the induction hypothesis, Falsifier has a dominant winning strategy for $\psi$. At $\varphi$, this constitutes Falsifier's winning strategy. Is his strategy then dominant? According to the truth table, the only possibilities for the conjuncts are true and false. Thus, they cannot be nonsense or super, otherwise $\varphi$ would not be false. Therefore, Dominator and Dictator have no winning strategies at $\varphi$. Thus, Falsifier's winning strategy is dominant.

\textsc{Conjunction for Dictator:} Let $\varphi = \psi \wedge \chi$. Assume that Dictator has a dominant winning strategy for $\varphi$. He follows his dominant strategy and chooses one of the conjuncts, say $\psi$. The induction hypothesis says that $\psi$ has the truth value super. Thus, by the the truth table of BH4 given in Figure~\ref{BH4truth}, $\varphi$ is super.

Conversely, if $\varphi$ is super, then some of the conjuncts has to be super. Dictator is allowed to make a move at a conjunction, and he chooses the super conjunct. The induction hypothesis says that Dictator has a dominant winning strategy for that conjunct. According to the game rules, Dictator's strategy dominates all. So, his dominant winning strategy at the conjunct remains to be a dominant strategy at $\varphi$.
\end{proof}

Following our methodology, different combinations of the above game rules can be constructed, yielding different logics with more or different infectious truth values. Furthermore, similar to what is suggested in \cite{ciu}, additional truth values, thus players, beyond the fourth can be introduced in a way that the dominant winning strategies form a linear order: Dominator dominates the classical players; Dictator dominates Dominator; a fifth player, say King, dominates Dictator and the rest etc. This procedure is rather straight-forward for countably many players. The semantic games seem to get more interesting once ${>}\omega$-many players are considered with a linear or branching order of dominant strategies.

\section{From BH3-Games to LP-Games}

In an earlier work, a game theoretical semantics for Graham Priest's Logic of Paradox (LP, for short) was given \cite{bas25}. In that work, the semantic games employed concurrent plays. At certain nodes, players engaged in concurrent plays where none were assumed to have dominant winning strategies. Consequently, those game rules produced significantly weaker correctness theorems for the game semantics for LP. It is then a natural question whether a game semantics for LP can be given using dominant strategies and if they entail a stronger correctness theorem. This would complement our overall aim of focusing on strategic dominance in game semantics for non-classical logics, and, at the same time, provide yet another game semantics for LP. This is our goal in this section.

The logic of paradox introduces a third truth value which was intended to represent paradoxical statements. The third truth value $P$ is called \emph{paradoxical} and requires its own player, which was called \emph{Astrolabe} \cite{bas25}. Astrolabe forces the game to an end with the truth value $P$, so he is the \emph{paradoxifier}. In semantic games for LP, conjunctions are for Falsifier and Paradoxifier and disjunctions are for Verifier and Paradoxifier. At negations, similar to BH3 games, Verifier and Falsifier switch roles, and Paradoxifier keeps his role. We reproduce the truth table for LP in the following.

\begin{figure}[h!]
\centering
\begin{tabular}{l|l}
& $\neg$ \\
\hline
$T$ & $F$ \\
$F$ & $T$ \\
$P$ & $P$ \\
\end{tabular} 
\qquad
\begin{tabular}{c|c|c|c}
$\wedge$ & $T$ & $P$ & $F$ \\
\hline
$T$ & $T$ & $P$ & $F$ \\
$P$ & $P$ & $P$ & $F$ \\
$F$ & $F$ & $F$ & $F$
\end{tabular}
\qquad
\begin{tabular}{c|c|c|c}
$\vee$ & $T$ & $P$ & $F$ \\
\hline
$T$ & $T$ & $T$ & $T$ \\
$P$ & $T$ & $P$ & $P$ \\
$F$ & $T$ & $P$ & $F$
\end{tabular}
\caption{\emph{The truth tables underlying} LP.}
{\label{LPtruth}}
\end{figure}

Notice that the very same game would work for the Strong Kleene System with minimal alterations as Strong Kleene and LP differ only on their designated truth values.

In \cite{bas25}, a game semantics for LP was suggested by allocating Astrolabe as the parallel player. He made moves, just as Dictator in BH3 games, along with Verifier or Falsifier. However, as the strategies and their strength were not properly considered in the aforementioned, the correctness theorems were not as strong.

\begin{thm}[\cite{bas25}]
In a GTS\textsuperscript{LP} semantic game $\Gamma_{\mathrm{LP}}(M, \varphi)$,

\begin{itemize}
\item Verifier has a winning strategy if $\varphi$ is true in $M$,
\item Falsifier has a winning strategy if $\varphi$ is false in $M$,
\item The pardoxifier has a winning strategy if $\varphi$ is paradoxical in $M$.
\end{itemize}\end{thm}

\begin{thm}[\cite{bas25}]{\label{thm-lp}}
In a GTS\textsuperscript{LP} semantic game $\Gamma_{\mathrm{LP}}(M, \varphi)$,

\begin{itemize}
\item If Verifier has a winning strategy, then  $\varphi$ is true in $M$,
\item If Falsifier has a winning strategy, then $\varphi$ is false in $M$,
\item If Paradoxifier has a winning strategy, but not the other players, then $\varphi$ is paradoxical in $M$.
\end{itemize}\end{thm}

These two theorems --especially the third bullet point in Theorem~\ref{thm-lp}-- seem to allude to some hierarchy between the player's strategies. In the sequel, we remedy this issue by imposing an order of domination over strategies.

We now give the new game rules for GTS\textsuperscript{LP} games as follows for $\Gamma_{\mathrm{LP}}(M, \varphi)$ with respect to players' roles.

\begin{itemize}
\item[($\rho^\mathrm{LP}_p$)]  If $\varphi$ is atomic, the game terminates, and Verifier wins if $\varphi$ is true, Falsifier wins if $\varphi$ is false, and Paradoxifier wins if $\varphi$ is paradoxical,
\item[($\rho^\mathrm{LP}_\neg$)]  if $\varphi = \neg \psi$, Verifier and Falsifier switch roles, Paradoxifier keeps his role, and the game continues as $\Gamma_{\mathrm{LP}}(M, \psi)$,
\item[($\rho^\mathrm{LP}_\wedge$)]  if $\varphi = \chi \wedge \psi$, Falsifier and Paradoxifier choose between $\chi$ and $\psi$ simultaneously,
\item[($\rho^\mathrm{LP}_\vee$)]  if $\varphi = \chi \vee \psi$, Verifier and Paradoxifier choose between $\chi$ and $\psi$ simultaneously,
\item[($\rho^\mathrm{LP}_s$)] Paradoxifier's strategy is strictly dominated by Verifier's and Falsifier's.
\end{itemize}

Similar to the case of BH3, Verifier and Falsifier will never have conflicting strategies. If either Verifier or Falsifier may have some conflicting strategy, it can only be with Paradoxifier where the pardoxifier's strategy is \emph{strictly dominated} by them.

It is important to notice that the difference in the truth tables for BH3 and LP (given in Figures~\ref{truthtable} and \ref{LPtruth}), can be accounted by the difference between the strategic dominance of players, as specified by game rules ($\rho_s$) and ($\rho^\mathrm{LP}_s$). Therefore, the only (game theoretical) difference between BH3 games and the LP games are whether the third player is the strictly dominant or the strictly dominated player. The following theorems summarizes our findings.

\begin{thm}
In a GTS\textsuperscript{LP} semantic game $\Gamma_{\mathrm{LP}}(M, \varphi)$, Verifier and Falsifier can never have winning strategies at the same time.
\end{thm}

\begin{proof}
The proof is very similar to that of Theorem~\ref{theo-dom}, hence skipped.	
\end{proof}

\begin{thm}{\label{thm-lp2}}
In a GTS\textsuperscript{LP} semantic game $\Gamma_{\mathrm{LP}}(M, \varphi)$,

\begin{itemize}
\item Verifier has a dominant winning strategy if and only if $\varphi$ is true in $M$,
\item Falsifier has a dominant winning strategy if and only if $\varphi$ is false in $M$,
\item Paradoxifier has a dominant winning strategy if and only if $\varphi$ is paradoxical in $M$.
\end{itemize}
\end{thm}

\begin{proof} 
The proof is by induction on $\varphi$ and very similar to the proof of Theorem~\ref{theo-correct}. We will only consider some interesting cases.

\textsc{Disjunction for Verifier:} Let $\varphi = \psi \vee \chi$. Assume that Verifier has a dominant winning strategy. Then, following his strategy he makes a choice, say $\psi$. His dominant strategy remains to be his dominant winning strategy at $\psi$. At the same time, Paradoxifier makes a choice. But whatever choice he makes, even if he has a winning strategy, by the game rules, it is strictly dominated by Verifier's strategy. Then, by the induction hypothesis, $\psi$ is true. According to the truth table for LP given in Figure~\ref{LPtruth}, $\varphi$ is true as well. 

Conversely, let $\varphi$ be true. Both Verifier and Paradoxifier make moves. According to the truth table, at least one of the disjuncts must be true, and Verifier chooses that. By the induction hypothesis, Verifier has a dominant winning strategy at that disjunct. Then, at $\varphi$, Verifier's dominant strategy is to choose the disjunct with the dominant winning strategy. Whatever choice Paradoxifier makes, his strategy is dominated. Thus, he cannot have a dominant strategy.

\textsc{Conjunction for Paradoxifier:} Let $\varphi = \psi \wedge \chi$. Assume that Paradoxifier has a dominant winning strategy. Since his strategy is always dominated by Verifier's and Falsifier's, this means that neither Verifier nor Falsifier has a winning strategy for $\varphi$. Now, both Paradoxifier and Falsifier are allowed to make a move. Falsifier does not have a winning strategy, which means that, by some other cases of the very theorem, neither of the conjuncts is false. If they were, then Falsifier may have a winning strategy, and by the game rules his strategy would be the dominant. Thus, Paradoxifier follows his winning strategy and makes a choice, say $\psi$, without loss of generality. According to the truth table, as we established that $\chi$ can only be paradoxical or true, $\varphi$ turns out to be paradoxical.

Conversely, let $\varphi = \psi \wedge \chi$ be paradoxical. According to the truth table for LP given in Figure~\ref{LPtruth}, conjuncts can be either true or paradoxical. Paradoxifier chooses the paradoxical conjunct. By the induction hypothesis, it is his winning strategy for that conjunct. And at $\varphi$, he constructs his winning strategy by choosing the said conjunct with a winning strategy. But, is Paradoxifier's winning strategy dominant? At conjunctions, Verifier is not allowed to make a move - if he was, he may have chosen the true conjunct if there was one. This would constitute his winning strategy, which, according to the game rules, would dominate Paradoxifier's. Therefore, at a paradoxical conjunct, Verifier is not allowed to make a move and Falsifier does not have a winning strategy. Hence, Paradoxifier's winning strategy is dominant.
\end{proof}

Compared to Theorem~\ref{thm-lp}, Theorem~\ref{thm-lp2} makes it precise what game theoretically corresponds to the condition that ``only paradoxifier has a winning strategy". We replaced this condition by dominated strategies. As such, the above result complements what has been presented in \cite{bas25}. It furthermore suggests that certain combinations of strategy dominance amongst the players may result in various interesting logical ideas. For example, in our work strategies strictly dominate some others. What is then the logic of those semantic games where strategies \emph{weakly} dominate? In the context of semantic games, weak strategy domination is an interesting issue.\footnote{We are thankful to the anonymous referee for pointing this out.} It may mean that the correctness theorems work only on one direction, failing to provide a bidirectional stronger result. We refer the interested reader to \cite{bas25}. 

Similar to the questions we raised for BH3/BH4, can we extend LP \emph{game theoretically} to a four-valued paraconsistent system? Such questions point out to various future work possibilities.

\section{Conclusion}

In this paper, we showed how non-classical game theory helps us understand the nuances of non-classical logics as well as certain game theoretical concepts.

To the best of our knowledge, this is the first game theoretical semantics suggested for Bochvar--Halld\'{e}n logics and infectious logics in general, relating infectiousness to strategy dominance. This proves that studying various combinations of dominant and dominated strategies is a fruitful direction to develop various new logics and study their model theory. A linear order of strategy dominance relates directly to certain solution concepts in game theory. Then the next question is to study a branching order of strategy dominance and how they relate to truth-coalitions in game semantics. 

Conversely, it is possible to inquire the opposite direction and develop logical methods for some other solution concepts and equilibria computation in game theory. Computing game theoretical equilibria is an important research direction in computer science. Incorporating logical elements to this methodology will certainly have immediate impact in aforementioned fields.

\paragraph{Acknowledgements} The feedback and suggestions of anonymous referees, which helped improve the paper, are much appreciated.

\bibliographystyle{eptcs} 
\bibliography{papers.bib}

\end{document}